\documentclass[a4paper,UKenglish,cleveref, autoref]{article}

\usepackage{multirow,array,booktabs}

\usepackage{amssymb}

\usepackage{amsthm}
\newtheorem{theorem}{Theorem}
\newtheorem{proposition}{Proposition}
\newtheorem{definition}{Definition}
\newtheorem{remark}{Remark}

\usepackage{url}

\title{LP-based algorithms for multistage minimization problems} 

\author{Evripidis Bampis$^1$, Bruno Escoffier$^{1,2}$, Alexander Kononov$^3$\\
$ $\\
$^1$ Sorbonne Universit\'e, CNRS, LIP6 UMR 7606,\\ 4 place Jussieu, 75005 Paris\\
$^2$ Institut Universitaire de France \\
$^3$ Sobolev Institute of Mathematics and Novosibirsk State University}

\date{}

\begin{document}

\maketitle

\begin{abstract}
  We consider a \emph{multistage} framework introduced recently (Eisenstat et al., Gupta et al., both in ICALP 2014), where given  a time horizon $t = 1, 2, \ldots,T$,  the input is a sequence of instances of a (static) combinatorial optimization problem $I_1,I_2,\ldots,I_T$, (one for each time step), and the goal is to find a sequence of solutions $S_1,S_2,\ldots,S_T$ (one for each time step) reaching  a tradeoff between the quality of the solutions in each time step and the stability/similarity of the solutions in consecutive time steps. For several polynomial-time solvable problems, such as Minimum Cost Perfect Matching, the multistage variant becomes hard to approximate (even for two time steps  for Minimum Cost Perfect Matching). In this paper, we study the multistage variants of some important discrete minimization problems (including Minimum Cut, Vertex Cover, Set Cover, Prize-Collecting Steiner Tree, Prize-Collecting Traveling Salesman). We focus on the natural question of whether linear-programming-based methods may help in developing good approximation algorithms in this framework. We first show that Min Cut remains polytime solvable in its multistage variant, and Vertex Cover remains 2-approximable, as particular case of a more general statement which easily follows from the work of (Hochbaum, EJOR 2002) on monotone and IP2 problems.
  Then, we tackle other problems and for this we introduce a new two-threshold rounding scheme, tailored for multistage problems. As a first application, we show that this rounding scheme gives a 2$f$-approximation algorithm for the multistage variant of the $f$-Set Cover problem, where each element belongs to at most $f$ sets. More interestingly, we are able to use our rounding scheme in order to propose 
  a 3.53-approximation algorithm for the multistage variant of the Prize-Collecting Steiner Tree problem, and a 3.034-approximation algorithm for the multistage variant of the Prize-Collecting Traveling Salesman problem.


\end{abstract}


\section{Introduction} In many applications, data are evolving with time and the systems have to be adapted  in order to take into account this evolution by providing near optimal solutions over the time. However, moving from a solution to a new one, may induce non-negligible transition costs. This cost may represent e.g. the cost of turning on/off the servers in a data center \cite{Albers17}, the cost of changing the quality level in video streaming \cite{Joseph}, or the cost for turning on/off nuclear plants in electricity production \cite{thesececile}. Various models and  algorithms have been proposed for modifying (re-optimizing) the current solution by making as few changes
as possible (see \cite{Anthony, Blanchard, Cohen, Gu, Megow, Nagarajan} and the references therein). 
In this paper, we follow a new trend, popularized by the works  of Eisenstat
et al. \cite{Eisenstat} and Gupta et al. \cite{Gupta}, known as  the \emph{multistage model}: Given a time horizon $t = 1, 2, \ldots,T$ and a  sequence of instances $I_1,I_2,\ldots,I_T$, (one for each time step),  the goal is to find a sequence of solutions $S_1,S_2,\ldots,S_T$ (one for each time step) optimizing the quality of the solution in each time step and the stability (transition cost or profit) between solutions in consecutive time steps. 
Surprisingly, even in the \emph{offline case}, where the sequence of instances is known in advance,  some  classic combinatorial optimization problems become much harder in the multistage model \cite{An, Bampis, Eisenstat, Gupta}. For instance,  
the Minimum Cost Perfect Matching problem, which is polynomially-time solvable in the one-step case, becomes hard to approximate even for bipartite graphs and for only two time steps \cite{Bampis}. 
In a more recent work, Fluschnik et al. \cite{Fluschnik} study multistage optimization problems from a complexity parameterized point of view, focusing on the multistage Vertex Cover problem.

In this article, we  focus on the \emph{complexity} and the \emph{approximability} of various multistage minimization problems including Min Cut, Vertex Cover, 
Prize-Collecting Steiner Tree and Prize-Collecting Traveling Salesman.
The central question that we address in this work is to what extend linear-programming based methods can be used in the multistage framework. Arguably one of the main techniques for the design of approximation algorithms for usual (static) problems, LP-based methods have already been fruitfully applied in the multistage setting for facility location problems in \cite{Eisenstat,An}.


\noindent
{\bf Our contribution.}
We first note that the multistage variants of monotone minimization problems (as defined in \cite{Hochbaum}) such as Min Cut, or IP2-non-monotone binarized minimization problems \cite{Hochbaum} such as Vertex Cover, remain monotone and IP2-non-monotone, respectively. Hence, the multistage variants of monotone static (one-step) problems can be solved in polynomial time, while the multistage variants of IP2-non-monotone static problems have the half-integrality property and consequently, they can be solved by a  2-approximation algorithm. Though obtained using simple arguments, we quote these results as they contrast with the several hardness results for multistage versions of classical problems such as Spanning Tree or Minimum Cost Perfect Matching. Indeed, Min Cut is, to the best of our knowledge, the first problem that is shown to remain polytime solvable in the multistage setting. Also, multistage Vertex Cover has the same approximation guarantee (2-approximation) as the static version (the ratio being tight under UGC).

Then, we focus on minimization problems which are neither monotone nor IP2. We introduce a new rounding scheme, called two-threshold rounding scheme, designed for multistage problems in the sense that it is able to take into account both transition costs and individual costs of solutions when rounding. We show a first application of this rounding scheme leading to a 2$f$-approximation algorithm for the multistage variant of the $f$-Set Cover problem (Set Cover where each element belongs to at most $f$ sets). As main results of this article, we then give a more involved use of this rounding scheme for a multistage variant of two prize-collecting problems: the multistage Prize-Collecting Steiner Tree, for which we obtain a 3.53-approximation algorithm, and the multistage Prize-Collecting Traveling Salesman problem, for which we obtain a 3.034-approximation algorithm.

\subsection{Related Work}

\emph{Multistage framework.} The multistage model considered in this paper has been introduced in Eisenstat
et al. \cite{Eisenstat} and Gupta et al. \cite{Gupta}. Eisenstat
et al. \cite{Eisenstat} studied the multistage version of facility location problems for which they proposed logarithmic approximation algorithms. An et al. \cite{An} obtained constant factor approximation algorithms for some related problems. Interestingly, these results are obtained using LP-based techniques, and in particular (randomized) roundings. These roundings are tailored for facility location problems, and thus quite different from the approach followed in this article.
 
 Gupta et al. \cite{Gupta} studied the multistage  Maintenance Matroid problem for
both the offline and the online settings. They presented a logarithmic approximation
algorithm for this problem, which includes as a special case a natural multistage version of the 
Spanning Tree problem. They also considered the online version of the problem and they devised an efficient randomized competitive algorithm against any oblivious adversary. The same paper also introduced the study of the multistage Minimum Cost Perfect Matching problem for which they proved that it is hard to approximate 
 even for a constant number of stages.  Bampis et al. \cite{Bampis} improved this negative result by showing that the problem is
hard to approximate even for bipartite graphs and for the case of two time steps. 
Olver et al.~\cite{Olver} studied a multistage version of the Minimum Linear Arrangement problem, which is related to a variant of the List Update problem~\cite{Sleator}, and provided a logarithmic lower bound for the online version and a polylogarithmic upper bound for the offline version.
Very recently, Fluschnik et al. \cite{Fluschnik} study the multistage Vertex Cover problem under a parameterized complexity perspective. They show that the multistage (offline) problem is computationally hard in fairly restricted cases and prove fixed-parameter tractability results for some natural parameterizations in some restricted cases.

Some multistage maximization problems have been studied as well. Let us quote the multistage Max-Min Fair Allocation
problem  \cite{Bampis+}, which is a multistage variant of the Santa Klaus problem \cite{Maxim}. Constant factor approximation algorithms are obtained for the off-line setting. 
More recently, in \cite{Bampis++}, a multistage variant of the Knapsack problem has been studied. A PTAS has been proposed and it has been proved that there is no FPTAS for the problem even in the case of two time steps, unless $P=NP$. 
In \cite{Bampis+++},  general techniques for a family of online multistage problems, called Subset Maximization, are developed and thereby characterizing the models (given by the type of data evolution and the type of similarity measure) that admit a constant-competitive online algorithm. 

Finally, let us mention the works by Buchbinder et
al. \cite{Buchbinder} and Buchbinder, Chen and Naor \cite{Buchbinder+} who considered also a multistage model and studied the relation between the online learning and competitive analysis frameworks, mostly for fractional optimization problems.\\ 

\noindent
\emph{Static framework.} 
In \cite{Hochbaum}, Hochbaum described  an easily recognizable class of integer programming problems, called \emph{monotone} problems, and proposed an algorithm that solves these problems in polynomial time, even in the case where the objective function is nonlinear. She also considered a large class of nonmonotone problems and she proposed a polynomial time algorithm that finds (fractional) optimal solutions that are half integral. These solutions can be used in order to devise 2-approximate solutions. In addition, the proposed 2-approximations are the best results that one can hope for these problems, unless there is a better approximation for Vertex Cover.  

For the Prize-Collecting Traveling Salesman problem that has been introduced in \cite{Balas}, Bienstock et al. \cite{Bienstock} proposed a 2.5-approximation algorithm. They have also proposed a 3-approximation algorithm for the Prize-Collecting Steiner Tree problem. Both algorithms use a rounding procedure of appropriate linear programming relaxations. Goemans and Williamson \cite{GW} improved these results providing a 2-approximation for both problems, based on the primal-dual scheme. Then, Archer et al. \cite{Archer}
devised a $(2-\epsilon)$-approximation algorithm for both problems. 
All these results hold for both the \emph{rooted} and \emph{unrooted} versions of the problems. This is true since there are approximation-preserving reductions from the  rooted to the unrooted version of the problems, and vice versa \cite{Archer}.

\section{Monotone and IP2 minimization problems}

\noindent {\bf Multistage problem and ILP formulation}

Let us consider a minimization problem where we want to minimize some linear cost function $\sum_i c_i x_i$ under some set of constraints defined as $Ax\geq b$ and $x_i\in \{0,1\}$. In the multistage model defined in \cite{Eisenstat,Gupta}, we are given a sequence of $T$ instances of the problem, i.e., $T$ cost functions $c_i^t$, $t=1,\dots,T$, and $T$ matrices $A_t$ and $b_t$, $t=1,\dots,T$. The goal is to find a sequence of feasible solutions $x^t,t=1,\dots,T$ (i.e., $A_t x^t \geq b_t$ and $x^t\in \{0,1\}^n$), so as to minimize an objective function which is the sum of:
\begin{itemize}
    \item The costs of individual solutions, i.e., $\sum_t \sum_i c^t_ix^t_i$;
    \item A transition cost: each time $x^t_i$ differs from $x^{t-1}_i$, this induces a cost $C^t$. Thus the global transition cost is $\sum_{t=1}^{T-1} C^t \sum_i |x^{t+1}_i-x^{t}_i|$.\\
\end{itemize}

This can be easily modeled as an ILP as follows.    

\begin{eqnarray*}
     \ \left \{ \begin{array}{ll}
    \min \ \sum_{t=1}^T\sum_{i=1}^n  c^t_{i}x^t_i+\sum_{t=1}^{T-1}C^t\sum_{i=1}^n (z^t_i+z'^t_i)   \\
    s.t. \left |
    \begin{array}{llllll}
    A_tx^t & \geq &b_t & \forall t\in \{1,\dots,T\}\\
    z^t_i & \geq & x^t_i-x^{t+1}_i & \forall t\in \{1,\dots,T-1\}, \forall{i} \in \{1,\dots,n\} \\
    z'^t_i & \geq & x^{t+1}_i - x^t_i & \forall t\in \{1,\dots,T-1\}, \forall{i} \in \{1,\dots,n\} \\
    x^t_{i}, z^t_i & \in & \{0,1\} \\
    \end{array}
    \right.
    \end{array} 
    \right.
    \end{eqnarray*}

As mentioned in introduction, several such multistage problems are hard to approximate, even when the underlying problem is polynomial. Here, we note that for some other problems the structure of the ILP formulation keep some properties in the multistage setting fitting the setting of Hochbaum~\cite{Hochbaum}, leading for instance to an exact (polytime) algorithm for multistage Min Cut or a 2-approximation for multistage Vertex Cover. 

\subsection{Minimum cut and monotone problems}

Let us first consider the (static) Min Cut problem, where we are given a graph $G=(V,E,s,p,c)$, two distinguished vertices $s$ and $p$ of $V$, and a function $c:E\rightarrow N$ associating a weight to each edge of the graph. A cut $(A,B)$ of $G$ is a partition of $V$ into $A$ and $B=V-A$ such that $s\in A$ and $p\in B$. The cost of the cut $(A,B)$ is $c(A,B)=\sum_{u \in A}\sum_{v \in B} c(u,v)$. A minimum cut in a graph is a cut whose cost is minimum over all cuts of the graph.

In the multistage version, as mentioned above, we consider that at every time step $t$, we are given a graph $G^t=(V,E^t,s,p,c^t)$, where the edge set and the cost of the edges may change over time. We denote by $E^t$ the edge set of the graph at time $t$ and by $c_{ij}^t$ the cost of edge $e=(i,j)\in E^t$. We denote by $v_i^t$ the vertex $v_i$ in $G^t$. There is a transition cost, $C^t$, for every vertex changing its partition-set from time $t$ to time $t+1$. 

We can solve the problem in polynomial time: indeed, the multistage problem on $T$ instances can be seen as a single (static) Min Cut problem on (nearly) $Tn$ vertices. To see this, we construct a new graph $G^{MS}$ in which we copy the $T$ graphs $G^1,G^2,\ldots,G^T$ and we connect  vertex $v_i^t$ to vertex $v_i^{t+1}$, for every $t=1,2,\ldots, T-1$ with an edge of weight $C^t$. We also add two vertices $\mathcal{S}$ and $\mathcal{T}$. We connect $\mathcal{S}$ (resp. $\mathcal{T}$) with every vertex $s^i$ (resp. $p^i$), for $i=1,2,\ldots,T$, with an edge $(\mathcal{S},s^i)$ (resp. $(p^i, \mathcal{T})$) of  weight  $\infty$. Then, it is sufficient to determine a minimum cut between $S$ and $\mathcal{T}$ in $G^{MS}$. Clearly, the cost of the minimum cut in $G^{MS}$ has a cost equal to the total cost, i.e. the sum of the cost of the cuts at every time step and of the transition cost between consecutive instances.

Another way to prove this result is to formulate the multistage Min Cut problem as an ILP and then prove that every optimal solution of the relaxation of this ILP is always an integer solution. In order to do that, let us first write the multistage Min Cut problem as an ILP where $x^t_i=1$ if vertex $v_i^t$ belongs to subset $A$, 0 otherwise. Also, let $z^t_i=1$ if $x^t_i \neq x^{t+1}_i$, 0 otherwise.

   \begin{eqnarray*}
     \ \left \{ \begin{array}{ll}
    \min \ \sum_{t=1}^T\sum_{e=(i,j)\in E^t} (f^t_{ij}+f'^t_{ij}) c^t_{ij}+\sum_{t=1}^{T-1}C^t\sum_{i=1}^n (z^t_i+z'^t_i)   \\
    s.t. \left |
    \begin{array}{llllll}
    f^t_{ij} & \geq &x_i^t-x_j^t & \forall t\in \{1,\dots,T\}, \forall{e=(i,j)} \in E^t\\
    f'^t_{ij} & \geq &x_j^t-x_i^t & \forall t\in \{1,\dots,T\}, \forall{e=(i,j)} \in E^t\\
    z^t_i & \geq & x^t_i-x^{t+1}_i & \forall t\in \{1,\dots,T-1\}, \forall{i} \in \{1,\dots,|V|\} \\
    z'^t_i & \geq & x^{t+1}_i - x^t_i & \forall t\in \{1,\dots,T-1\}, \forall{i} \in \{1,\dots,|V|\} \\
    x^1_{\mathcal{S}}=1,  x^T_{\mathcal{T}}=0 \\
    x^t_{i}, z^t_i \in \{0,1\} \\
    \end{array}
    \right.
    \end{array} 
    \right.
    \end{eqnarray*}

In order to show that the relaxation of this IPL is integral, we use the result of Hochbaum \cite{Hochbaum} that a class of minimization integer programming problems, known as {\em monotone}, are solvable in polynomial time. The problems in this class are characterized by constraints of the form $ax-by \leq c+z$, where $a,b\geq 0$ and the variable $z$ appears only in that constraint. The direction of the inequality is not important and the coefficients $a$ and $b$
can be any real number if $b\geq 1$. The objective function is unrestricted, but the function of $z$ must be convex. It is not difficult to verify that the ILP of the multistage Min Cut problem is monotone. It is important here to notice that a minimization problem which is monotone in the static case continues to be monotone in the multistage framework. This is a consequence of the fact that the constraints containing the transition variables and the way the transition
cost is added in the objective function induce a  monotone integer program. 

\begin{proposition}
Multistage Min Cut, as well as the multistage version of any monotone  minimization problem, is polynomial time solvalbe. 
\end{proposition}

Note that the same holds in the more general case where the transition costs also depend on the vertex, i.e., there is a cost $C^t_i$ for changing decision about vertex $v_i$ from time $t$ to $t+1$.

\subsection{Vertex Cover and IP2 problems}

As it is well known, the relaxation of the standard formulation of the minimum Vertex Cover problem has the semi-integrality property, meaning that any extremal solution  has coordinates in $\{0,1/2,1\}$. This has been generalized by Hochbaum \cite{Hochbaum} in a class of nonmonotone minimization integer programs, known as IP2. These optimization problems have constraints of the type $ax+by \leq z+c$ without restriction in the sign of $a$ and $b$. Some NP-hard problems can be modeled this way, but in the case where all coefficients in the constraint matrix are in $\{-1,0,1\}$, the IP2 problem is said to be {\em binarized}. For binarized IP2 integer programs, Hochbaum \cite{Hochbaum} propose a polynomial time algorithm that provides half integral (fractional) optimal solutions for the continuous relaxation. When the half integral solution can be rounded to a feasible integral solution, this gives a 2-approximation polynomial time algorithm. As for monotone integer programs, it is easy to see that the multistage variant of an IP2 problem belongs to the class of IP2 problems and hence, a 2-approximation algorithm exists for all these problems, including for instance multistage Vertex Cover.

\begin{proposition}
Multistage Vertex Cover, as well as the multistage version of any minimization binarized IP2 integer problems, has a 2-approximation algorithm. 
\end{proposition}
As previsouly, the same holds in the model where the transition cost also depend on the vertex $v_i$/variable $x_i$.

Note that for Vertex Cover this is the best we can hope for since even the static version is hard to approximate within ratio $2-\epsilon$ under UGC \cite{KhotR08}.


\section{Rounding scheme}

We now tackle problems which are neither monotone nor IP2. We propose a new rounding method, called 2-threshold rounding scheme, which allows to take into account both individual costs of solutions and transition costs in the multistage setting. We first show in Section~\ref{sec:fsc} an easy application of this rounding to the $f$-Set Cover problem (Set Cover where each element appears in at most $f$ sets), leading to a $2f$-approximation algorithm for the multistage version of this problem. We then show in Sections~\ref{sec:st} and \ref{sec:tsp} how to use the 2-threshold rounding scheme for a multistage version of the two problems Prize-Collecting Steiner Tree and Prize-Collecting Traveling Salesman. We obtain respectively a 3.53- and a 3.034-approximation algorithm.\\

Let us now introduce the rounding scheme. A classical rounding scheme for a static problem expressed as an ILP is as follows: starting from an optimal solution $\hat{x}=(\hat{x}_1,\dots,\hat{x}_n)$ of a continuous relaxation, fix a threshold $\alpha\in (0,1)$, and fix $x_i=1$ if $\hat{x}_i \geq \alpha$, and  $x_i=0$ otherwise. This gives for instance an $f$-approximation algorithm for $f$-Set Cover. However, this rounding is not suitable for multistage problem as it may induce very large transition cost: if $\hat{x}_i^t=\alpha$ and $\hat{x}_i^{t+1}=\alpha-\epsilon$, then $|x^{t+1}_i-x^t|=1$ while $|\hat{x}^{t+1}_i-\hat{x}^t|=\epsilon$.

To overcome this, we introduce the following 2-threshold rounding scheme.
\begin{definition}(2-threshold rounding scheme)
Let $x=(x^1,\dots,x^T)$ a sequence of $T$ values in $[0,1]$, and $\alpha,\beta$ two parameters with $1\geq \alpha \geq \beta \geq 0$. Then \texttt{RS}$(x,\alpha,\beta)$ is the sequence $y=(y^1,\dots,y^T)$ of $T$ values in $\{0,1\}$ defined as:
\begin{itemize}
	\item If $x^t\geq \alpha$ then $y^t=1$.
	\item If $x^t\leq \beta$ then $y^t=0$.
	\item  Consider a maximal interval $[t_1,t_2]$ with $\beta<x^t< \alpha$ for all $t\in[t_1,t_2]$. Then: if simultaneously (1) $t_1=1$ or $x^{t_1-1} \geq \alpha$, and (2) $t_2=T$ or $x^{t_2+1}\geq \alpha$, then fix $y^t=1$ for all $t\in [t_1,t_2]$. Otherwise ((1) or (2) is not verified) then fix $y^t=0$ for all $t\in [t_1,t_2]$.
\end{itemize} 
\end{definition}

As we will see, these two thresholds allow to bound both transition costs and individual costs of solutions when rounding.

\subsection{Rounding scheme for $f$-Set Cover}\label{sec:fsc}

We consider the Set Cover problem where we are given a ground set $\mathcal{C}=\{c_1,\dots,c_n\}$, and a collection $\mathcal{S}=\{S_1,\dots,S_m\}$ of subsets of $\mathcal{C}$. Each set $S_i$ has a nonnegative weight $w_i$. The goal is to find a subcollection of $\mathcal{S}$ of minimum weight such that each element of the ground set is in at least one of the chosen sets. The $f$-Set Cover problem corresponds to instances where each element is in at most $f$ sets. 
Note that for any $f$, Set Cover is not $(f-1-\epsilon)$-approximable if $P\neq NP$ \cite{DinurGKR05}, and not $(f-\epsilon)$-approximable under UGC \cite{KhotR08}. In the sequel we show that a general multistage version of the problem is $2f$-approximable.\\

In the multistage version, we consider that a set $S_i$ may change over time (it may contain $c_j$ at time $t$ but not at time $t+1$), and its weight may also change: we denote $S^t_i$ the set $S_i$ at time $t$, and $w^t_i$ its weight. We get a penalty $p_i$ if we change our decision about set $S_i$ between time $t$ and $t+1$. 

We write this problem as an ILP, where $x^t_i=1$ if set $S_i$ is taken at time $t$, 0 otherwise. $z^t_i=1$ if $x^t_i=x^{t+1}_i$, 0 otherwise.

   \begin{eqnarray*}
     \ \left \{ \begin{array}{ll}
    \min \ \sum_{t=1}^T\sum_{i=1}^m w^t_i x^t_{i}+\sum_{t=1}^{T-1}\sum_{i=1}^m p_i z^t_i   \\
    s.t. \left |
    \begin{array}{llllll}
    \sum_{i: c_j\in S^t_i} x^t_{i} & \geq &1 & \forall t\in \{1,\dots,T\}, \forall{c_j} \in \mathcal{C}\\
    z^t_i & \geq & x^t_i-x^{t+1}_i & \forall t\in \{1,\dots,T-1\}, \forall{i} \in \{1,\dots,m\} \\
    z^t_i & \geq & x^{t+1}_i - x^t_i & \forall t\in \{1,\dots,T-1\}, \forall{i} \in \{1,\dots,m\} \\
    x^t_{i}, z^t_i \in \{0,1\} \\
    \end{array}
    \right.
    \end{array} 
    \right.
    \end{eqnarray*}
The first constraint ensures that each element is covered at each time step. In $f$-Set Cover instances, this constraint involves at most $f$ variables. Notice that we can also allow the ground set to change.\\ 

Let us consider the following algorithm \texttt{RS-MSC} (Rouding Scheme for multistage $f$-Set Cover):
\begin{enumerate}
    \item Compute an optimal (fractional) solution  $(\hat{x},\hat{z})$ of the relaxation of the ILP formulation.
    \item For each set $S_i$, apply the 2-threshold rounding scheme on $\hat{x_i}$ with parameters $\alpha= \frac{1}{f}$ and $\beta=\frac{1}{2f}$, i.e. let $\tilde{x}_i=$\texttt{RS}$(\hat{x}_i,\frac{1}{f},\frac{1}{2f})$. This defines a solution where $S_i$ is taken at time $t$ iff $\tilde{x}_i^t=1$.
\end{enumerate}







\begin{theorem}
\texttt{RS-MSC} is a $2f$-approximation algorithm for the $f$-Set Cover problem. 
\end{theorem}
\begin{proof}
We first show that the solution is feasible: for each element $c$ and each $t$, there exists $S_i$ containing $c$ such that $\hat{x}_i^t\geq \frac{1}{f}$ (feasibility constraint). Then the rounding scheme fix $\tilde{x}_i^t=1$, i.e., $S_i$ is taken at time $t$ and $c$ is covered. 

Now let us show the approximation ratio. First notice that if $\hat{x}_i^t\leq \frac{1}{2f}$ then $\tilde{x}_i^t=0$, so we have for any $i,t$ $\tilde{x}_i^t\leq 2f \hat{x}_i^t$. Then $\sum_t\sum_i w^t_i\tilde{x}^t_i\leq 2f \sum_t  \sum_i w^t_i \hat{x}^t_i$.

It remains to bound the transition cost of the solution. In the computed solution, $\tilde{x}^t_i$ jumps (once) from 0 to 1 only if $\hat{x}^t_i\leq \frac{1}{2f}$,  $\hat{x}^{t'}_i\geq \frac{1}{f}$ for some $t'>t$ and   $\frac{1}{2f}< \tilde{x}^{t''}_i< \frac{1}{f}$ for all $t<t''<t'$ (or $t'=t+1$). But then, the global transition cost of $\hat{x}_i$ on the period between $t$ and $t'$ is at least $p_i(\hat{x}^{t'}_i-\hat{x}^{t}_i)\geq \frac{p_i}{2f}$, while the transition cost of  $\tilde{x}_i$ is $p_i$ (one jump from 0 to 1).

Similarly, $\tilde{x}^t_i$ jumps (once) from 1 to 0 only if $\hat{x}^t_i\geq \frac{1}{f}$,  $\hat{x}^{t'}_i\leq \frac{1}{2f}$ for some $t'>t$ and   $\frac{1}{2f}< \hat{x}^{t''}_i< \frac{1}{f}$ for all $t<t''<t'$ (or $t'=t+1$). Again, the global transition cost of $\hat{x}_i$ on the period between $t$ and $t'$ is at least $p_i|\hat{x}^{t'}_i-\hat{x}^{t}_i|\geq \frac{p_i}{2f}$, while the transition cost of  $\tilde{x}_i$ is $p_i$. Globally, the transition cost of the computed solution is at most $2f$ times the one of the optimal (continuous) solution.
\end{proof}



\subsection{Prize-Collecting Steiner Tree}\label{sec:st}


In this section, we show how to use the 2-threshold  rounding scheme as an important ingredient in order to derive approximation algorithms for the Prize-Collecting Steiner Tree problem. In the (static)  Prize-Collecting Steiner Tree problem we have: a graph $G=(V,E)$, a root $r\in V$, each edge has a cost $c(e)\geq 0$, each vertex $v$ has a penalty $\pi(v)\geq 0$. We assume the graph to be complete (w.l.o.g., by putting large weights). Given a subset of edges $F\subseteq E$, we denote $C(F)$ the set of vertices connected to the root by this set of edges. The value $z(F)$ associated to $F$ is $z(F)=\sum_{e\in F} c(e) +\sum_{v\not\in C(F)}\pi(v)$. Given that costs are nonnegative, there always exists an optimal solution which forms a tree (with the root inside).

We formulate the (static) problem with the following ILP, where $x_e=1$ if $e$ is taken in the solution, and $s_v=1$ if $v$ is connected to the root. We denote by $\mathcal{P}_v$ the set of sets of vertices $S$ with $v\in S$ and $r\not \in S$. $\delta(S)$ is the set of edges with one endpoint in $S$ and one endpoint outside $S$.

   \begin{eqnarray*}
     \ \left \{ \begin{array}{ll}
    \min \ \sum_{e \in E} c(e) x_{e}+\sum_{v\in V} \pi(v)(1-s_v)   \\
    s.t. \left |
    \begin{array}{llllll}
    \sum_{e \in \delta(S)} x_{e} & \geq &s_v & \forall{v} \in V, \forall S \in \mathcal{P}_v \\
    x_{e},s_v \in \{0,1\} \\
    \end{array}
    \right.
    \end{array} 
    \right.
    \end{eqnarray*}

The constraint ensures that if we choose to connect $v$ to the root ($s_v=1$),  $\forall S \in \mathcal{P}_v$ at least one edge goes outside $S$. \\

We consider a multistage version where we have a transition cost induced by modifying our decision about vertex $v$ between time steps $t$ and $t+1$. Namely, there is a cost $w_v$ if we connect $v$ to the root at time step $t$, but not at time step $t+1$, or vice-versa. We formulate the problem as an ILP as follows: $x^t_e=1$ if edge $e$ is taken at time $t$, $s^t_v=1$ if $v$ is connected to $r$ at time $t$, $z^t_v=1$ if $s^t_v\neq s^{t+1}_v$ (transition cost).

   \begin{eqnarray*}
     \ \left \{ \begin{array}{ll}
    \min \ \sum_{t=1}^T\sum_{e \in E} c^t(e) x^t_{e}+\sum_{t=1}^T\sum_{v\in V} \pi^t(v)(1-s^t_v)+\sum_{t=1}^{T-1}\sum_{v\in V} w_v z^t_v   \\
    s.t. \left |
    \begin{array}{llllll}
    \sum_{e \in \delta(S)} x^t_{e} & \geq &s^t_v & \forall t\in \{1,\dots,T\}, \forall{v} \in V, \forall S \in \mathcal{P}_v \\
    z^t_v & \geq & s^t_v-s^{t+1}_v & \forall t\in \{1,\dots,T-1\}, \forall{v} \in V \\
    z^t_v & \geq & s^{t+1}_v - s^t_v & \forall t\in \{1,\dots,T-1\}, \forall{v} \in V \\
    x^t_{e},s^t_v, z^t_v \in \{0,1\} \\
    \end{array}
    \right.
    \end{array} 
    \right.
    \end{eqnarray*}

Let us consider the algorithm \texttt{RS-MPCST} (Rouding Scheme for Multistage Prize-Collecting Steiner Tree) which works as follows.

\begin{enumerate}
	\item Find an optimal solution $(\hat{x},\hat{s},\hat{z})$ to the relaxation of the ILP where variables are in $[0,1]$.
	\item Apply the rounding scheme \texttt{RS} to each sequence of variables $\hat{s}_u=(s^1_u,\dots,s^T_u)$ with parameters $\alpha$ and $\beta$ (to be specified), let $\tilde{s}_u=(\tilde{s}^1_u,\dots,\tilde{s}^T_u)$ be the corresponding (integer) values. At this point, we have decided for each time $t$ which vertices we want to connect to the root $r$ (those for which $\tilde{s}^t_u=1$).
	\item For each $t=1,\dots,T$, apply the algorithm \texttt{ST-Algo} of \cite{GW} for the Steiner Tree problem on the instance whose costs on edges are the ones at time $t$, and the set of terminals is $\{v:\tilde{s}^t_v=1\}$ (those we have to connect). This gives the set of edges chosen at time $t$.
\end{enumerate}

The algorithm \texttt{ST-Algo} has the following property \cite{GW}: it computes a (intregal) solution the value of which is at most $2$ times the value of a (feasible) solution of a dual formulation of the problem. More precisely, consider the following ILP formulation of an instance of Steiner Tree, where $T$ is the set of terminals, and $\mathcal{P}$ is the set of sets $S$ such that $S\cap T\neq \emptyset$ and $r\not\in S$.

   \begin{eqnarray*}
     (ILP-ST) \  \left \{ \begin{array}{ll}
    \min \ \sum_{e \in E} c(e) x_{e}   \\
    s.t. \left |
    \begin{array}{llllll}
    \sum_{e \in \delta(S)} x_{e} & \geq &1 & \forall S \in \mathcal{P} \\
    x_{e} \in \{0,1\} \\
    \end{array}
    \right.
    \end{array} 
    \right.
    \end{eqnarray*}

Associating a variable $y_S$ to each $S \in \mathcal{P}$, the following LP $(D-ST)$ is the dual of the relaxation (LP-ST) of (ILP-ST) where we relax $x_{e} \in \{0,1\}$ as $x_e\geq 0$.

   \begin{eqnarray*}
     (D-ST) \left \{ \begin{array}{ll}
    \max \ \sum_{S \in \mathcal{P}}  y_{S}   \\
    s.t. \left |
    \begin{array}{llllll}
    \sum_{S\in \mathcal{P}:e \in \delta(S)} x_{e} & \geq &c(e) & \forall e \in E \\
    y_{e}\geq 0 \\
    \end{array}
    \right.
    \end{array} 
    \right.
    \end{eqnarray*}

Then \texttt{ST-Algo} computes a feasible solution $y_S$ of the dual and an integral solution $\hat{x}$ of the primal such that $\sum_{e}c(e)\hat{x}_e\leq 2 \sum_S y_S$. 

\begin{theorem}
With parameters $\alpha=3/4$ and $\beta=1/2$, \texttt{RS-MPCST} is a $4$-approximation algorithm.
\end{theorem} 

\begin{proof}
We note that the LP can be solved in polynomial time by the Ellipsoid method.
The separation problem for the constraints $\sum_{e \in \delta(S)} x^t_{e} \geq s^t_v$ is reduced to the Min Cut problem and can be solved in polynomial time. 

By solving a Steiner Tree problem in step 3, we connect at each time step the vertices we chose (during step 2) to connect to $r$, so the solution is clearly feasible. 

Let $(\tilde{x},\tilde{s})$ be the computed solution ($\tilde{z}$ being immediately deduced from $\tilde{s}$). First, we note that in the rounding scheme, if $\hat{s}^t_i\geq \alpha$ then $\tilde{s}^t_i= 1$, so for any $i,t$ $(1-\tilde{s}^t_i)\leq \frac{1-\hat{s}^t_i}{1-\alpha}$, and we can bound the loss on the cost of not connecting some vertices:

\begin{equation}\label{STeq1}
	\sum_{t=1}^T\sum_{v\in V} \pi^t(v)(1-\tilde{s}^t_v) \leq \frac{\sum_{t=1}^T\sum_{v\in V} \pi^t(v)(1-\hat{s}^t_v)}{1-\alpha}
\end{equation}

Now, as for the case of Set Cover, a variable  $\tilde{s}^t_u$ jumps (once) from 0 to 1 only if $\hat{s}^t_u\leq \beta$,  $\hat{s}^{t'}_u\geq \alpha$ for some $t'>t$ and   $\beta< \hat{s}^{t''}_u< \alpha$ for all $t<t''<t'$ (or $t'=t+1$). But then, the global transition cost of $\hat{s}^t_u$ on this period between $t$ and $t'$ is at least $w_u(\hat{s}^{t'}_v-\hat{s}^t_v)\geq w_u(\alpha-\beta)$, while the transition cost of  $\tilde{s}^t_u$ is $w_u$. The same argument holds for the jumps of $\overline{s}^t_u$ from 1 to 0. Then we can bound the loss on the transition costs: 

\begin{equation}\label{STeq2}
\sum_{t=1}^{T-1}\sum_{v\in V} w_v \tilde{z}^t_v \leq \frac{\sum_{t=1}^{T-1}\sum_{v\in V} w_v \hat{z}^t_v}{\alpha-\beta}
\end{equation}

Now we have to deal with the cost of taking edges. We consider some time step $t$. Since $(\hat{x},\hat{s},\hat{z})$ is feasible,  for all $v \in V$, all $S \in \mathcal{P}_v$:
$$
\sum_{e \in \delta(S)} \hat{x}^t_{e}  \geq \hat{s}^t_v
$$
Now, since $\tilde{s}^t_v=0$ if $\hat{s}^t_v\leq \beta$, we have $\tilde{s}^t_v\leq \frac{\hat{s}^t_v}{\beta}$. Then
$$
\sum_{e \in \delta(S)} \beta\hat{x}^t_{e}  \geq \beta\hat{s}^t_v \geq \tilde{s}^t_v
$$

Considering the formulation (LP-ST) of the Steiner Tree problem where the set of terminals is $T=\{v:\tilde{s}^t_v\}=1$, we have that  $\sum_{e \in \delta(S)} \beta\hat{x}^t_{e}  \geq 1$ for all $S$ containing at least one terminal (one $v$ with $\tilde{s}^t_v=1$), and not containing $r$. In other words,  $\beta\hat{x}^t_{e}$ is a (continuous) feasible solution for (LP-ST) (with  $T=\{v:\tilde{s}^t_v\}=1$). 
By duality, its value is at least the value of any dual feasible solution $y_S$: 

$$\sum_{e \in E} c^t(e) \beta \hat{x}^t_{e}\geq \sum_{S}y_S$$ 

Since \texttt{ST-Algo} computes a solution $\hat{x}$ of cost at most $2\sum_{S}y_S$, we get, for all $t$:
\begin{equation}\label{STeq3}
\sum_{e \in E} c^t(e) \tilde{x}^t_{e} \leq \frac{2}{\beta}\sum_{e \in E} c^t(e) \hat{x}^t_{e}
\end{equation}

In all, the computed solution has ratio at most $\frac{1}{1-\alpha}$ for the cost of not connecting vertices, ratio at most $\frac{1}{\alpha-\beta}$ for the transition costs, and ratio at most $\frac{2}{\beta}$ for the cost of edges. By choosing $\beta=\frac{1}{2}$ and $\alpha=\frac{3}{4}$, we get ratio at most $4$ in each case.
\end{proof}

\subsubsection*{Improvement via (de)randomization} 
Following an approach used for the classical Prize-Collecting Steiner Tree problem (see \cite{WilliamsonBook}) we give here a randomized algorithm for the multistage Prize-collecting Steiner Tree problem leading to a better (expected) ratio. We then show how it can be derandomized. 

The randomized algorithm is the same as \texttt{RS-MPCST} except that we choose $\alpha$ at random uniformly from the range $[\gamma,1],$ where $\gamma = e^{-\frac13} > 0$, and we set $\beta = \frac23 \alpha.$ 
The probability density function for a uniform random variable over $[\gamma,1]$ is the constant $1/(1-\gamma).$ 
We have
\begin{equation}\nonumber
\mathbb{E}[\sum_{t=1}^T\sum_{v\in V} \pi^t(v)(1-\tilde{s}^t_v)]=\mathbb{E}[\sum_{t=1}^T(\sum_{v|\hat{s}^t_v<\gamma} \pi^t(v)(1-\tilde{s}^t_v)
+\sum_{v|\hat{s}^t_v\geq \gamma} \pi^t(v)(1-\tilde{s}^t_v)]
\end{equation}
\begin{equation}\nonumber
= \sum_{t=1}^T\sum_{v|\hat{s}^t_v<\gamma} \pi^t(v) + \sum_{t=1}^T\sum_{v|\hat{s}^t_v\geq \gamma}\left(\mathbb{E}[\pi^t(v) Pr(\overline{s}^t_v=0)] \right)
\end{equation}
\begin{equation}
= \sum_{t=1}^T\left( \sum_{v|\hat{s}^t_v<\gamma} \frac{\pi^t(v)(1-\hat{s}^t_v)}{(1-\hat{s}^t_v)} + \sum_{v|\hat{s}^t_v\geq \gamma} \frac{\pi^t(v)(1-\hat{s}^t_v)}{1-\gamma} \right)
\leq \frac{1}{1-\gamma}\sum_{t=1}^T\sum_{v\in V} \pi^t(v)(1-\hat{s}^t_v).
\end{equation}
From (\ref{STeq2}) and linearity of expectations, we obtain:
\begin{equation}\nonumber
\mathbb{E}[\sum_{t=1}^{T-1}\sum_{v\in V} w_v \tilde{z}^t_v] \leq \sum_{t=1}^{T-1}\sum_{v\in V} w_v \hat{z}^t_v\mathbb{E}[\frac{1}{\alpha-\beta}] 
= \sum_{t=1}^{T-1}\sum_{v\in V} w_v \hat{z}^t_v\mathbb{E}[\frac{3}{\alpha}]
\end{equation}
\begin{equation}
=\left(\frac{3}{1-\gamma}\ln\frac{1}{\gamma}\right) \sum_{t=1}^{T-1}\sum_{v\in V} w_v \hat{z}^t_v
=\frac{1}{1-\gamma} \sum_{v\in V} w_v \hat{z}^t_v.
\end{equation}
Finally, from (\ref{STeq3}) we have 
$$
\mathbb{E}[\sum_{e \in E} c^t(e) \overline{x}^t_{e}] \leq \mathbb{E}[\frac{2}{\beta}\sum_{e \in E} c^t(e) \hat{x}^t_{e}]=\sum_{e \in E} c^t(e) \hat{x}^t_{e}] \mathbb{E}[\frac{3}{\alpha}]
=\frac{1}{1-\gamma}\sum_{e \in E} c^t(e) \hat{x}^t_{e}.
$$
Thus, the expected cost of the obtained solution  is no greater than $\frac{1}{1-e^{-\frac13}}OPT \leq 3.53OPT.$\\

Now, let us deal with derandomization. The random selection of $\alpha$ determines  the splitting of the values $\hat{s}_u^t$ into three sets ${\cal V}^{+}(\alpha)=\{\hat{s}_u^t | \hat{s}_u^t \geq \alpha \} $,
${\cal V}^{-}(\alpha)=\{\hat{s}_u^t | \hat{s}_u^t \leq \beta \},$ and ${\cal V}(\alpha)=\{\hat{s}_u^t | \beta < \hat{s}_u^t < \alpha \}.$ Since $\beta = \frac23 \alpha,$ 
any possible value of $\alpha$ corresponds to a splitting of the set of values into three sets. It is clear that there are at most $(|V|T)^2$ such partitions ($\alpha$ or $\beta$ taking one of the at most $|V|t$ values $\hat{s}_u^t$). Let us call \texttt{I-RS-MPCST} (for improved \texttt{RS-MPCST}) the corresponding derandomized algorithm. Following the above discussion, we get: 

\begin{theorem}
\texttt{I-RS-MPCST} is a (deterministic) 3.53-approximation algorithm. 
\end{theorem}

\begin{remark}
The derandomization can be done in a more efficient way, by considering no more than $2|V|T$ different values of $\alpha.$ 
For each $s_u^t > 0$, we set $\alpha(u,t)=s_u^t $ and $\alpha'(u,t)=\frac32 s_u^t.$ 
Let ${\cal A}=\{\alpha(u,t)|s_u^t > 0, u \in V, t=1, \ldots,T\} \cup \{\alpha'(u,t)|s_u^t > 0, u \in V, t=1, \ldots,T\}.$ 
Consider an arbitrary $\alpha \in [e^{-\frac13},1].$ We show that there exists $\alpha' \in {\cal A}$ such that ${\cal V}^{+}(\alpha)={\cal V}^{+}(\alpha'),$
${\cal V}^{-}(\alpha)={\cal V}^{-}(\alpha'),$ and ${\cal V}(\alpha)={\cal V}(\alpha').$ Let $\eta_0=\min\{s_u^t|s_u^t\in{\cal V}\}$ and 
$\eta_1=\min\{s_u^t|s_u^t\in{\cal V}^{+}\}.$  We have $\frac32 \eta_0 \in {\cal A}$ and $\eta_1 \in {\cal A}.$
If 
\begin{equation} \label{dereq1}
\eta_1 - \alpha \leq \frac32(\eta_0 - \beta),
\end{equation} 
then we set $\alpha' = \eta_1,$ so $\beta'=\frac23 \eta_1.$  It is clear that ${\cal V}^{+}(\alpha)={\cal V}^{+}(\alpha').$
From (\ref{dereq1}) we have
$$ \eta_0 \geq \frac23(\eta_1-\alpha)+\beta \geq \frac23(\eta_1-\alpha)+\frac23 \alpha = \frac23 \eta_1 = \beta'.$$
Hence, ${\cal V}^{-}(\alpha)={\cal V}^{-}(\alpha'),$ and ${\cal V}(\alpha)={\cal V}(\alpha').$ If
\begin{equation} \label{dereq2}
\eta_1 - \alpha > \frac32(\eta_0 - \beta),
\end{equation} 
then we set $\alpha' =\frac32 \eta_0.$ Then $\beta' = \eta_0$ and we obtain that ${\cal V}^{-}(\alpha)={\cal V}^{-}(\alpha').$
From (\ref{dereq2}) we have
$$ \eta_1  > \frac32(\eta_0 - \beta) +  \alpha = \frac32 \eta_0 = \alpha'.$$
It follows that ${\cal V}^{+}(\alpha)={\cal V}^{+}(\alpha'),$ and ${\cal V}(\alpha)={\cal V}(\alpha').$
\end{remark}

\subsection{Prize-Collecting Traveling Salesman}\label{sec:tsp}

In this section we consider the Prize-Collecting Traveling Salesman problem. We have a complete graph $G=(V,E)$,  a depot $r\in V$, 
each edge has a cost $c(e)$, each vertex $v$ has a penalty $\pi(v).$ We assume that the vertex $r$ must be in the tour, i.e.  the tour starts and ends at vertex $r$. The edge costs are assumed to satisfy the triangle inequality.
In the Price Collecting Traveling Salesman problem, it is required to find a tour that visits a subset of the vertices 
such that the length of the tour plus the sum of penalties of all vertices not in the tour  is as small as possible. 
We consider the multistage version of the problem in which the costs of edges and the penalties may change over time. 
Additionally, we have a transition cost induced by modifying our decision about vertex $v$ between time $t$ and $t+1$. Namely, we pay
a cost $w_v$ if we visit $v$ at time step $t$ but not at time step $t+1$, or vice-versa. 

We adapt the ILP for the Prize-Collecting Travelling Salesman problem introduced in~\cite{Bienstock}. 
Let $G'=(V',E')$ be the complete graph  resulting from $G$ by adding a dummy vertex $r'.$ We set $c^t(\{r,r'\})=0$ and $c^t(\{v,r'\})=c^t(\{v,r\})$ for all $v\in V \setminus \{r\}.$
Let $x^t_e=1$ if edge $e$ is in the tour at time $t$ and zero otherwise, $s^t_v=1$ if $v$ is in the tour  at time $t$ and zero otherwise, 
$z^t_v=1$ if $s^t_v\neq s^{t+1}_v$ (transition cost).

\begin{eqnarray*}
    ILP-MPCTSP \ \left \{ \begin{array}{ll}
    \min \ \sum_{t=1}^T\sum_{e \in E'} c^t(e) x^t_{e}+\sum_{t=1}^T\sum_{v\in V} \pi^t(v)(1-s^t_v)+\sum_{t=1}^{T-1}\sum_{v\in V} w_v z^t_v   \\
    s.t. \left |
    \begin{array}{llllll}
    \sum_{e \in \delta(\{v\})} x^t_{e} & = &2s^t_v & \forall t\in \{1,\dots,T\}, \forall{v} \in V' \\
    \sum_{e \in \delta(S)} x^t_{e} & \geq &2s^t_v & \forall t\in \{1,\dots,T\}, \forall{v} \in V, \forall S \subset V \\
    & & &   \mbox{such that} \, \, \,  |S\cap \{r',v\}|=1 \\
    z^t_v & \geq & s^t_v-s^{t+1}_v & \forall t\in \{1,\dots,T-1\}, \forall{v} \in V \\
    z^t_v & \geq & s^{t+1}_v - s^t_v & \forall t\in \{1,\dots,T-1\}, \forall{v} \in V \\
    x^t_{e},s^t_v, z^t_v \in \{0,1\} \\
    \end{array}
    \right.
    \end{array} 
    \right.
    \end{eqnarray*}
    
    Let us consider the algorithm \texttt{RS-MPCTSP} (Rounding Scheme for multistage Prize-Collecting TSP).

\begin{enumerate}
	\item Find an optimal solution $(\hat{x},\hat{s},\hat{z})$ to the relaxation of ILP-MPCTSP where variables are in $[0,1]$.
	\item Apply the rounding scheme \texttt{RS} to each sequence of variables $\hat{s}_u=(\hat{s}^1_u,\dots,\hat{s}^T_u)$ with parameters $\alpha$ and $\beta$ (to be specified), 
	let $\tilde{s}_u=(\tilde{s}^1_u,\dots,\tilde{s}^T_u)$ be the corresponding (integer) values. 
	\item Let $R_t=\{v:\tilde{s}^t_v=1\}.$ For each $t=1,\dots,T$, construct a Traveling Salesman tour through all vertices in $R_t$ using Christofides' algorithm \cite{Christophides}. 
\end{enumerate}

\begin{theorem}
With parameters $\alpha=5/7$ and $\beta=3/7$, \texttt{RS-MPCTSP} is a $3.5$-approximation algorithm.
\end{theorem} 
\begin{proof}
Let $L^{\ast}(R_t)$ be the length of an optimal tour through the subset of vertices $R_t$ and
$L(R_t)$ be the length of the tour through the subset of vertices $R_t$ produced by  Christofides' algorithm. 
The following linear program was introduced by Bienstock et al. \cite{Bienstock}.

\begin{eqnarray*}
    LP-BGSW(t) \ \left \{ \begin{array}{ll}
    \min \ \sum_{e \in E'} c^t(e) x^t_{e}  \\
    s.t. \left |
    \begin{array}{lllll}
    \sum_{e \in \delta(S)} x^t_{e} & \geq &2&   \forall S \subset V' : R_t\cap S\ne \emptyset, \, R_t\cap(V' \setminus S) \ne \emptyset\\
    x^t_{e} & \geq &0 & \forall e \in E'\\
    \end{array}
    \right.
    \end{array} 
    \right.
    \end{eqnarray*}

Let $\bar{x}$ be an optimal solution of LP-BGSW and $Z(\bar{x})$ be the cost of this solution.
Bienstock et al. proved that $Z(\bar{x})$ coincides with the well-known lower bound obtained by Held and Karp \cite{Held}. 
It follows that $Z(\bar{x}) \leq L^{\ast}(R_t)$ and $L(R_t) \leq \frac32 Z(\bar{x}).$ The second inequality follows from 
a similar result obtained for the Held and Karp lower bound obtained in \cite{Shmoys,Wolsey}.  Thus, we have 
$$ L(R_t) \leq \frac32 \sum_{e \in E'} c^t(e) \bar{x}^t_{e}.$$

Define new vectors $\tilde{x}^t_{e}=\frac{\hat{x}^t_{e}}{\beta}$ for all $e \in E$ and $t=1,\ldots,T.$ 
Consider some time step $t.$ We now show that $\tilde{x}^t_{e}$ is feasible for  LP-BGSW(t).
To prove that the first constraints  are satisfied, we consider any $S \subset V'$ such that $v \in R_t\cap S$ and $r' \in R_t \setminus S.$
In this case, we have
$$ \sum_{e \in \delta(S)} \tilde{x}^t_{e} = \frac{1}{\beta} \sum_{e \in \delta(S)} \hat{x}^t_{e} \geq \frac{2\hat{s}^t_v}{\beta} \geq 2.$$
The first inequality follows from the feasibility of $(\hat{x},\hat{s},\hat{z}).$ According to the rounding scheme \texttt{RS} $\hat{s}^t_v \geq \beta$
for all $v \in R_t$ and the last inequality also holds. From optimality of $\bar{x}$ we obtain 
$$  \sum_{e \in E'} c^t(e) \bar{x}^t_{e} \leq  \sum_{e \in E'} c^t(e)  \tilde{x}^t_{e}.$$
Hence, 
\begin{equation} \label{eq-TSP}
L(R_t) \leq \frac32 \sum_{e \in E'} c^t(e) \bar{x}^t_{e} \leq \frac32   \sum_{e \in E'} c^t(e)  \tilde{x}^t_{e} 
\leq \frac{3}{2\beta} \sum_{e \in E'} c^t(e)  \hat{x}^t_{e}.
\end{equation} 
 
 Repeating the reasoning of the previous section we obtain that
 the computed solution has a ratio of at most $\frac{1}{1-\alpha}$ for the cost of the set of the unserved  vertices, and a ratio of at most $\frac{1}{\alpha-\beta}$ for the transition costs.
 Finally, by choosing $\beta=\frac{3}{7}$ and $\alpha=\frac{5}{7}$, we get a ratio of at most $3.5$ in each case.
     \end{proof}

\subsubsection*{Improvement via (de)randomization}
We use the same technique as the one of multistage Prize-Collecting Steiner Tree in order to improve the ratio. First we consider a randomized algorithm where we choose $\alpha$ uniformly from the range $[\gamma,1],$ where $\gamma = e^{-\frac25} > 0$ and set $\beta = \frac35 \alpha.$ 
The probability density function for a uniform random variable over $[\gamma,1]$ is the constant $1/(1-\gamma).$ 
As previously, we have:
\begin{equation}
\mathbb{E}[\sum_{t=1}^T\sum_{v\in V} \pi^t(v)(1-\overline{s}^t_v)] \leq \frac{1}{1-\gamma}\sum_{t=1}^T\sum_{v\in V} \pi^t(v)(1-\hat{s}^t_v).
\end{equation}
From (\ref{STeq2}) and linearity of expectations, we obtain:
\begin{equation}\nonumber
\mathbb{E}[\sum_{t=1}^{T-1}\sum_{v\in V} w_v \overline{z}^t_v] \leq \sum_{t=1}^{T-1}\sum_{v\in V} w_v \hat{z}^t_v\mathbb{E}[\frac{1}{\alpha-\beta}] 
= \sum_{t=1}^{T-1}\sum_{v\in V} w_v \hat{z}^t_v\mathbb{E}[\frac{5}{2\alpha}]
\end{equation}
\begin{equation}
=\left(\frac{5}{2(1-\gamma)}\ln\frac{1}{\gamma}\right) \sum_{t=1}^{T-1}\sum_{v\in V} w_v \hat{z}^t_v
=\frac{1}{1-\gamma} \sum_{v\in V} w_v \hat{z}^t_v.
\end{equation}
Finally from (\ref{eq-TSP}), we have: 
$$
\mathbb{E}[\sum_{e \in E} c^t(e) \overline{x}^t_{e}] \leq \mathbb{E}[\frac{3}{2\beta}\sum_{e \in E} c^t(e) \hat{x}^t_{e}]=\sum_{e \in E} c^t(e) \hat{x}^t_{e}] \mathbb{E}[\frac{5}{2\alpha}]
=\frac{1}{1-\gamma}\sum_{e \in E} c^t(e) \hat{x}^t_{e}.
$$
Thus, the expected cost of the obtained solution  is no greater than $\frac{1}{1-e^{-\frac25}}OPT \leq 3.034OPT.$\\

The derandomization can be done in a very similar way as for Prize-Collecting Steiner Tree, leading to the algorithm  \texttt{I-RS-MPCTSP} and the following result.
\begin{theorem}
\texttt{I-RS-MPCTSP} is a (deterministic) 3.034-approximation algorithm. 
\end{theorem}

\bibliography{biblio}

\begin{thebibliography}{10}

\bibitem{Albers17}
Susanne Albers.
\newblock On energy conservation in data centers.
\newblock In {\em ACM Symposium on Parallelism in Algorithms and Architectures
  ({SPAA})}, pages 35--44, 2017.

\bibitem{An}
Hyung{-}Chan An, Ashkan Norouzi{-}Fard, and Ola Svensson.
\newblock Dynamic facility location via exponential clocks.
\newblock {\em {ACM} Trans. Algorithms}, 13(2):21:1--21:20, 2017.

\bibitem{Anthony}
Barbara~M. Anthony and Anupam Gupta.
\newblock Infrastructure leasing problems.
\newblock In {\em Conference on Integer Programming and Combinatorial
  Optimization (IPCO)}, pages 424--438, 2007.

\bibitem{Archer}
Aaron Archer, MohammadHossein Bateni, MohammadTaghi Hajiaghayi, and Howard~J.
  Karloff.
\newblock Improved approximation algorithms for prize-collecting steiner tree
  and {TSP}.
\newblock {\em {SIAM} J. Comput.}, 40(2):309--332, 2011.

\bibitem{Balas}
Egon Balas.
\newblock The prize collecting traveling salesman problem.
\newblock {\em Networks}, 19(6):621--636, 1989.

\bibitem{Bampis}
Evripidis Bampis, Bruno Escoffier, Michael Lampis, and Vangelis~Th. Paschos.
\newblock Multistage matchings.
\newblock In {\em Scandinavian Symposium and Workshops on Algorithm Theory
  ({SWAT})}, pages 7:1--7:13, 2018.

\bibitem{Bampis+}
Evripidis Bampis, Bruno Escoffier, and Sasa Mladenovic.
\newblock Fair resource allocation over time.
\newblock In {\em International Conference on Autonomous Agents and MultiAgent
  Systems (AAMAS)}, pages 766--773, 2018.

\bibitem{Bampis+++}
Evripidis Bampis, Bruno Escoffier, Kevin Schewior, and Alexandre Teiller.
\newblock Online multistage subset maximization problems.
\newblock In {\em European Symposium on Algorithms ({ESA})}, pages 11:1--11:14,
  2019.

\bibitem{Bampis++}
Evripidis Bampis, Bruno Escoffier, and Alexandre Teiller.
\newblock Multistage knapsack.
\newblock In {\em Mathematical Foundations in Computer Science ({MFCS})}, pages
  22:1--22:14, 2019.

\bibitem{Maxim}
Nikhil Bansal and Maxim Sviridenko.
\newblock The santa claus problem.
\newblock In Jon~M. Kleinberg, editor, {\em Proceedings of the 38th Annual
  {ACM} Symposium on Theory of Computing, Seattle, WA, USA, May 21-23, 2006},
  pages 31--40. {ACM}, 2006.

\bibitem{Bienstock}
Daniel Bienstock, Michel~X. Goemans, David Simchi-Levi, and David~P.
  Williamson.
\newblock A note on the prize collecting traveling salesman problem.
\newblock {\em Mathematical Programming}, 59:413--420, 1993.

\bibitem{Blanchard}
Nicolas~K. Blanchard and Nicolas Schabanel.
\newblock Dynamic sum-radii clustering.
\newblock In {\em International Conference and Workshops on Algorithms and
  Computation (WALCOM)}, pages 30--41, 2017.

\bibitem{Buchbinder+}
Niv Buchbinder, Shahar Chen, and Joseph Naor.
\newblock Competitive analysis via regularization.
\newblock In {\em {ACM-SIAM} Symposium on Discrete Algorithms (SODA)}, pages
  436--444, 2014.

\bibitem{Buchbinder}
Niv Buchbinder, Shahar Chen, Joseph Naor, and Ohad Shamir.
\newblock Unified algorithms for online learning and competitive analysis.
\newblock {\em Math. Oper. Res.}, 41(2):612--625, 2016.

\bibitem{Christophides}
Nicos Christofides.
\newblock Worst-case analysis of a new heuristics for the traveling salesman
  problem.
\newblock In {\em Report 388, Graduate School of Industrial Administration,
  Carnegie-Mellon University (Pittsburg, PA).}, 1976.

\bibitem{Cohen}
Edith Cohen, Graham Cormode, Nick~G. Duffield, and Carsten Lund.
\newblock On the tradeoff between stability and fit.
\newblock {\em {ACM} Trans. Algorithms}, 13(1):7:1--7:24, 2016.

\bibitem{DinurGKR05}
Irit Dinur, Venkatesan Guruswami, Subhash Khot, and Oded Regev.
\newblock A new multilayered {PCP} and the hardness of hypergraph vertex cover.
\newblock {\em {SIAM} J. Comput.}, 34(5):1129--1146, 2005.

\bibitem{Eisenstat}
David Eisenstat, Claire Mathieu, and Nicolas Schabanel.
\newblock Facility location in evolving metrics.
\newblock In {\em International Colloquium on Automata, Languages, and
  Programming ({ICALP})}, pages 459--470, 2014.

\bibitem{Fluschnik}
Till Fluschnik, Rolf Niedermeier, Valentin Rohm, and Philipp Zschoche.
\newblock Multistage vertex cover.
\newblock In {\em International Symposium on Parameterized and Exact
  Computation ({IPEC})}, 2019.

\bibitem{GW}
Michel~X. Goemans and David~P. Williamson.
\newblock A general approximation technique for constrained forest problems.
\newblock {\em {SIAM} J. Comput.}, 24(2):296--317, 1995.

\bibitem{Gu}
Albert Gu, Anupam Gupta, and Amit Kumar.
\newblock The power of deferral: Maintaining a constant-competitive steiner
  tree online.
\newblock {\em {SIAM} J. Comput.}, 45(1):1--28, 2016.

\bibitem{Gupta}
Anupam Gupta, Kunal Talwar, and Udi Wieder.
\newblock Changing bases: Multistage optimization for matroids and matchings.
\newblock In {\em International Colloquium on Automata, Languages, and
  Programming ({ICALP})}, pages 563--575, 2014.

\bibitem{Held}
Michael Held and Richard~M. Karp.
\newblock The travelling salesman problem and minimum spanning trees.
\newblock {\em Operations Research}, 18:1138--1162, 1970.

\bibitem{Hochbaum}
Dorit~S. Hochbaum.
\newblock Solving integer programs over monotone inequalities in three
  variables: {A} framework for half integrality and good approximations.
\newblock {\em European Journal of Operational Research}, 140(2):291--321,
  2002.

\bibitem{Joseph}
Vinay Joseph and Gustavo de~Veciana.
\newblock Jointly optimizing multi-user rate adaptation for video transport
  over wireless systems: Mean-fairness-variability tradeoffs.
\newblock In {\em {IEEE} International Conference on Computer Communications
  ({INFOCOM})}, pages 567--575, 2012.

\bibitem{KhotR08}
Subhash Khot and Oded Regev.
\newblock Vertex cover might be hard to approximate to within 2-epsilon.
\newblock {\em J. Comput. Syst. Sci.}, 74(3):335--349, 2008.

\bibitem{Megow}
Nicole Megow, Martin Skutella, Jos{\'{e}} Verschae, and Andreas Wiese.
\newblock The power of recourse for online {MST} and {TSP}.
\newblock {\em {SIAM} J. Comput.}, 45(3):859--880, 2016.

\bibitem{Nagarajan}
Chandrashekhar Nagarajan and David~P Williamson.
\newblock Offline and online facility leasing.
\newblock {\em Discrete Optimization}, 10(4):361--370, 2013.

\bibitem{Olver}
Neil Olver, Kirk Pruhs, Rene Sitters, Kevin Schewior, and Leen Stougie.
\newblock The itinerant list-update problem.
\newblock In {\em Workshop on Approximation and Online Algorithms (WAOA)},
  pages 310--326, 2018.

\bibitem{thesececile}
C\'ecile Rottner.
\newblock {\em Combinatorial Aspects of the Unit Commitment Problem}.
\newblock PhD thesis, Sorbonne Universit\'e, 2018.

\bibitem{Shmoys}
David Shmoys and David Williamson.
\newblock Analyzing the held-karp tsp bound: A monotonicity property with
  applications.
\newblock {\em Information Processing Letters}, 35(6):281--285, 1988.

\bibitem{Sleator}
Daniel~Dominic Sleator and Robert~Endre Tarjan.
\newblock Amortized efficiency of list update and paging rules.
\newblock {\em Commun. {ACM}}, 28(2):202--208, 1985.

\bibitem{WilliamsonBook}
David~P. Williamson and David~B. Shmoys.
\newblock {\em The Design of Approximation Algorithms}.
\newblock Cambridge University Press, New York, NY, USA, 1st edition, 2011.

\bibitem{Wolsey}
Laurence~A. Wolsey.
\newblock Heuristic analysis, linear programming and branch and bound.
\newblock {\em Mathematical Programming Study}, 13:121--134, 1980.

\end{thebibliography}

\end{document}